\documentclass[runningheads]{llncs}
\usepackage{graphicx}
\usepackage{listings}
\usepackage[english]{babel}
\usepackage[utf8x]{inputenc}
\usepackage{amsmath}
\usepackage{graphicx}
\usepackage[colorinlistoftodos]{todonotes}
\usepackage{amsmath,amssymb} 
\usepackage{array,multirow}
\usepackage{makecell}
\usepackage{array,rotating}
\usepackage{amssymb,diagbox}
\usepackage{tikz}
\usepackage{rotating}
\usetikzlibrary{automata,positioning}
\usepackage{caption}
\usepackage{multicol}
\usepackage{float}
\usepackage{chngpage}
\usepackage{url}
\usetikzlibrary{decorations.pathmorphing}
\tikzset{snake it/.style={decorate, decoration=snake}}

\usepackage{hyperref}

\newtheorem{thm}{Theorem}
\newtheorem{df}[thm]{Definition}

\newtheorem{rem}[thm]{Remark}

\usepackage[all]{xy}
\usepackage{float}
\usepackage{wrapfig}

\definecolor{gray}{rgb}{0.4,0.4,0.4}
\definecolor{white}{rgb}{1,1,1}
\definecolor{red}{rgb}{0.8,0,0}

\newcommand{\bbn}{\mathbb N}

\begin{document}
	\title{Planar digraphs for automatic complexity\thanks{
			This work was partially supported by
			a grant from the Simons Foundation (\#315188 to Bj\o rn Kjos-Hanssen).
			We are indebted to Jeff Shallit and Malik Younsi for helpful comments.
		}
	}
	\author{
		Achilles A.~Beros \and
		Bj{\o}rn Kjos-Hanssen\orcidID{0000-0002-6199-1755} \and
		Daylan Kaui Yogi
	}
	\institute{
		University of Hawai`i at M{\=a}noa, Honolulu HI 96822, USA
		\email{\{beros,bjoernkh,dkyogi64\}@hawaii.edu}
	}
	\authorrunning{A.~Beros, B.~Kjos-Hanssen et al.}
	\maketitle
	\begin{abstract}
		We show that
		the digraph of a nondeterministic finite automaton witnessing the automatic complexity of a word
		can always be taken to be planar.
		In the case of total transition functions studied by Shallit and Wang, planarity can fail.

		Let $s_q(n)$ be the number of binary words $x$ of length $n$ having nondeterministic automatic complexity $A_N(x)=q$.
		We show that $s_q$ is eventually constant for each $q$ and that the eventual constant value of $s_q$ is computable.

	\end{abstract}

	\section{Introduction}
		Automatic complexity, introduced by Shallit and Wang \cite{MR1897300}, is an automata-based and length-conditional analogue of
		Sipser's CD complexity \cite{Sipser:1983:CTA:800061.808762} which is in turn a computable analogue of the noncomputable
		Kolmogorov complexity.
		The nondeterministic case was taken up by Hyde and Kjos-Hanssen \cite{MR3386523},
		who gave a table of the number of words of length $n$ of a given complexity $q$ for $n\le 23$.
		The numbers in the table suggested (see Table \ref{tab:prop}) that the number may be eventually constant for each fixed $q$.
		Here we establish that that is the case (Theorem \ref{thm:total}), and show that the limit is computable (in exponential time).
		Moreover, we narrow down the possible automata that are needed to witness nondeterministic automatic complexity:
		they must have planar digraphs, in fact their digraphs are trees of cycles in a certain sense.

		We recall our basic notion.
		 \begin{df}[{\cite{MR1897300}}]
			The \textbf{nondeterministic automatic complexity} $A_{N}(x)$ of a word $x$ is
			the minimal number of states of a nondeterministic finite automaton $M$ (without $\epsilon$-transitions) accepting $x$ such that
			there is only one accepting path in $M$ of length $|x|$.
		 \end{df}

	\section{Automatic complexity as chains of trees of lumps}

		Consider the version of automatic complexity where the transition functions are not required to be total.\footnote{
			Whether determinism is required is not important in the following,
			but in the nondeterministic case we assume we require there to be only one accepting path, as usual.
		}
		Then we claim that the digraphs representing the witnessing automata are planar, in fact they are ``trees of cycles''.
		As an example, for the word $0^510^51^6010^3$, we have the following witnessing automaton:
		\[
			\begin{tikzpicture}[shorten >=1pt,node distance=1.5cm,on grid,auto]
				\node[state,initial] (q_0) {$q_0$}; 
				\node[state] (q_1) [below=of q_0] {$q_1$};
				\node[state] (q_2) [below right=of q_1] {$q_2$};
				\node[state,accepting] (q_3) [above right=of q_2] {$q_3$};
				\node[state] (q_4) [above=of q_3] {$q_4$};
				\node[state] (q_5) [above left=of q_4] {$q_5$};
				\node[state] (q_6) [above=of q_5] {$q_6$};
				\node[state] (q_7) [above right=of q_6] {$q_7$};
				\node[state] (q_8) [below right=of q_7] {$q_8$};
				\node[state] (q_9) [below=of q_8] {$q_9$};
				\node[state] (q_10) [above left=of q_6] {$q_{10}$};
				\path[->] 
					(q_0) edge [swap] node {$0$} (q_1)
					(q_1) edge [swap] node {$0$} (q_2)
					(q_2) edge [swap] node {$0$} (q_3)
					(q_3) edge [swap] node {$0$} (q_4)
					(q_4) edge [swap] node {$0$} (q_5)
					(q_5) edge [swap] node {$1$} (q_6)
					(q_6) edge [swap] node {$0$} (q_7)
					(q_7) edge [swap] node {$0$} (q_10)
					(q_10)edge [swap] node {$0$} (q_6)
					(q_7) edge node {$0$} (q_8)
					(q_5) edge node {$1$} (q_0)
					(q_8) edge [swap] node {$1$} (q_9)
					(q_9) edge [bend right,swap] node {$1$} (q_8)
					(q_9) edge node {$1$} (q_4)
					;
			\end{tikzpicture}
		\]
		To explain this, first let us say that a cycle is a sequence of states that starts and ends with the same state.
		Let us say that a lump is the automaton whose transitions come from a given cycle.
		So if a cycle is repetitive, like 3456734567345673, then it generates the same lump as just 345673.

		Consider the sequence of states visited during processing of a unique accepted word x of length n.
		Let us call the first visited state 0, the next distinct state 1, and so on.
		(So for example the permitted state sequences of length 3 are only 000, 001, 010, 011, 012.)

		Then the state sequence starts $0,1,\dots,q,q+1,\dots,q$ where $q$ is the first state that is visited twice.
		Now the claim is that there will never, at a later point in the state sequence, be a transition (an edge) $q_1$,$q_2$ such that
		$q_2$ occurs within the lump generated by the cycle $q,q+1,\dots,q$ and such that
		the transition $q_1,q_2$ does not occur in that lump. Indeed, otherwise our state sequence would start
		\[
		0,1,\dots,\underbrace{q,\dots,q_2}_{\text{first}},\dots,\underbrace{q,\dots,q_1,q_2}_{\text{second}}
		\]
		and then there is a second accepting path of the same length where the first and second segments are switched.

		Consequently, the path can only return to states that are not yet in any lumps.
		This leaves only two choices whenever we decide to create a new edge leading to a previously visited state:

		Case 1. Go back to a state that was first visited after the last completed lump so far seen, or
		Case 2. Go back to a state that was first visited at some earlier time, before some of the lumps so far seen started
		(and in general after some of them were complete).

		This gives a tree of lumps where each new lump either
		(Case 1) creates a new sibling for the previous lump, or
		(Case 2) creates a new parent for a final segment of the so far seen top-level siblings.
		In this tree of lumps, only the leaves (the lumps that are not anybody's parents)
		can be traversed more than once by the uniquely accepted path of length n.

		So if the first lump created is $l_1$ then next we can have two cases:
		\[
		(l_1,\quad l_2) \tag{Case 1}
		\]
		\[
			l_1\to l_2\tag{Case 2}
		\]
		In Case 1, $l_1$ and $l_2$ are siblings ordered from first to second.
		In Case 2, $\to$ denotes \emph{is a child of}, which by definition is the same as \emph{sub-digraph}.
		Now for the third lump $l_3$, we have only the following possibilities:
		\[
			(l_1, \quad l_2,\quad l_3)\tag{Subcase 1.1}
		\]
		\[
			(l_1,\quad l_2\to l_3)\tag{Subcase 1.2}
		\]
		\[
		(l_1, l_2)\to l_3\tag{Subcase 1.3}
		\]
		\[
		(l_1\to l_2,\quad l_3)\tag{Subcase 2.1}
		\]
		\[
		l_1\to l_2\to l_3\tag{Subcase 2.2}
		\]
		In Subcase 1.2, $l_1$ and $l_3$ are siblings and $l_2$ is a child of $l_3$.
		In Subcase 1.3, $l_3$ is a common parent of $l_1$ and $l_2$.
		In Subcase 2.1, $l_3$ is a new sibling for $l_2$, and $l_2$ still has $l_1$ as its child.
		In Subcase 2.2, $l_3$ is a parent of $l_2$.

		For instance, the state sequence 01234567345673456720 has the structure of Subcase 2.2, with
		$l_1$ being the lump generated from 345673,
		$l_2$ being generated from 23456734567345672, and
		$l_3$ being generated from the whole sequence 01234567345673456720. The corresponding automaton is shown in an online tool.\footnote{ 
		\url{http://math.hawaii.edu/wordpress/bjoern/complexity-of-0001111011110111111/}
		}
		Using this planarity result, we are able to increase the speed of our algorithm for calculating $A_N(x)$.
		Consequently, we have been able to extend the string length in our computations from $n=23$ to $n=25$.
		The number of maximally complex binary words of a given length are shown in Table \ref{max}.
		A similar table for $n\le 23$ was given in \cite{MR3386523}.
		\begin{table}
			\centering
			\begin{tabular}{r|r|r|r|r}
				$n$&\#&			$2^n$&		\%complex&	$2^n$-\#\\
				\hline
				0&	1&			1&			100.00\%&	0\\
				1&	2&			2&			100.00\%&	0\\
				\hline
				2&	2&			4&			50.00\%&	2\\
				3&	6&			8&			75.00\%&	2\\
				\hline
				4&	8&			16&			50.00\%&	8\\
				5&	24&			32&			75.00\%&	8\\
				\hline
				6&	30&			64&			46.88\%&	34\\
				7&	98&			128&		76.56\%&	30\\
				\hline
				8&	98&			256&		38.28\%&	158\\
				9&	406&		512&		79.30\%&	106\\
				\hline
				10&	344&		1,024&		33.59\%&	680\\
				11&	1,398&		2,048&		68.26\%&	650\\
				\hline
				12&	1,638&		4,096&		39.99\%&	2,458\\
				13&	5,774&		8,192&		70.48\%&	2,418\\
				\hline
				14&	5,116&		16,384&		31.23\%&	11,268\\
				15&	23,018&		32,768&		70.25\%&	9,750\\
				\hline
				16&	22,476&		65,536&		34.30\%&	43,060\\
				17&	86,128&		131,072&	65.71\%&	44,944\\
				\hline
				18&	89,566&		262,144&	34.17\%&	172,578\\
				19&	351,250&	524,288&	67.00\%&	173,038\\
				\hline
				20&	375,710&	1,048,576&	35.83\%&	672,866\\
				21&	1,461,670&	2,097,152&	69.70\%&	635,482\\
				\hline
				22&	1,539,164&	4,194,304&	36.70\%&	2,655,140\\
				23&	5,687,234&	8,388,608&	67.80\%&	2,701,374\\
				\hline
				24&	6,814,782&	16,777,216&	40.62\%&	9,962,434\\
				25&	24,031,676&	33,554,432&	71.62\%&	9,522,756\\
				\hline
				26&	27,782,964&	67,108,864&	41.40\%&	39,325,900\\
				27& 97,974,668&	134,217,728&73.00\%&	36,243,060\\
				\hline
			\end{tabular}
			\caption{Lengths $n$,
				number of words of length $n$ of maximal $A_N(x)$,
				$2^n$,
				percentage of maximally complex words,
				number of non-maximally complex words.
			}\label{max}
		\end{table}
	
	\section{The asymptotic number of words of given complexity}
	In this section, we examine the asymptotic behavior of the number of words with automatic complexity $q$ for a fixed $q \in \bbn$. 
	\begin{df}
			A binary word $x$ is \emph{right inextendible} if $A_{N}(x)<A_{N}(x0)$ and $A_{N}(x)<A_{N}(x1)$.
		\end{df}
		Inextendibility is closely related to volatility of the automatic complexity, as examined in the Complexity Option Game \cite{COG}.
		The number and proportion of right-inextendible words of length $n$ and complexity $q$ can be examined using
		an online database \cite{lookup} and is shown in Table \ref{tab:prop} for small $q$ and $n$.

		A basic procedure in our results will be the counting of periodic words, since a cycle containing a periodic word can be shortened and
		an automaton containing such a cycle will not be optimal.
		\begin{df}
			A word $x$ is \emph{periodic} if there exists a subword $y\not=x$ and an integer $n$ such that
			\[
				\underbrace{yyy\cdots y}_n=x.
			\]
		\end{df}
		A non-periodic word \cite{mrx05} is also called a primitive word and one starting with 0, in our setting, is called a \emph{Lyndon word} \cite{MR0064049}.
		\begin{df}[{\cite{MR0376360}}]
			Let $n$ be a positive integer with
			$\omega(n)$ denoting the number of distinct prime factors of $n$ and
			$\Omega(n)$ denoting the total number of prime factors (i.e., with repetition) of $n$.
			The M\"obius function $\mu$ is defined as
			\[
				\mu(n):=\begin{cases}
					          (-1)^{\omega(n)\bmod 2} &\text{ if }\Omega(n)=\omega(n),\\
					\phantom{+}0 &\text{ if }\Omega(n)>\omega(n).
				\end{cases}
			\]
		\end{df}
		\begin{thm}[{\cite{mrx05}}]\label{thm:upbs}
			The number of unique periodic binary words of length $n$ is given by $Z(0)=0$ and for $n\ge 1$,
			\[
				Z(n)= 2^{n}-\sum_{d|n}\mu\left(\frac{n}{d}\right)\cdot 2^{d}.
			\]
		\end{thm}
		Recall that a \emph{necklace} is an equivalence class of non-periodic words under cyclic rotation.
		Thus, for instance, $\{0011,0110,1100,1001\}$ is a necklace.
		Theorem \ref{thm:upbs} is a restatement of the following classical result.
		\begin{thm}[Witt's Formula \cite{MR1581553}]
			The number of necklaces of binary words of length $n$ is
			\[
				\frac1n \sum_{d|n}\mu\left(\frac{n}{d}\right)\cdot 2^{d}.
			\]
		\end{thm}
		\begin{sidewaystable}
			\centering
			\begin{tabular}{|r|c|c|c|c|c|c|c|c|}
				\hline
				\diagbox{$n$}{$q$} & 3 & 4 & 5 & 6 & 7 & 8 & 9 & 10\\
				\hline
				22 & 8/20 & 28/58 & 86/164 & 322/502 & 1288/2846 & 6594/16024 & 44922/94732 & 220544/451368\\
				\hline
				21 & 8/20 & 28/58 & 98/176 & 292/496 & 1318/3168 & 8472/18720 & 52178/108042 & 266760/504794\\
				\hline
				20 & 8/20 & 28/58 & 86/164 & 238/430 & 1478/3814 & 11670/23328 & 54990/115896 & 278696/529148\\
				\hline
				19 & 8/20 & 28/58 & 86/164 & 402/582 & 2380/4996 & 12312/26542 & 78892/410668 & 134578/351250\\
				\hline
				18 & 8/20 & 28/58 & 110/188 & 356/598 & 2070/5692 & 14456/29990 & 68288/36024 & 0/0\\
				\hline
				17 & 8/20 & 28/58 & 104/200 & 262/514 & 2850/7102 & 20516/37042 & 30486/86128 & \\
				\hline
				16 & 8/20 & 28/58 & 80/164 & 536/752 & 2908/7738 & 14230/34320 & 0/22476 & \\
				\hline
				15 & 8/20 & 28/58 & 148/226 & 578/908 & 3338/8530 & 7524/23018  &  & \\
				\hline
				14 & 8/20 & 28/58 & 112/244 & 774/1270 & 4442/9868 & 0/5116 &  & \\
				\hline
				13 & 8/20 & 28/58 & 120/250 & 1396/2076 & 1736/5774 &  &  & \\
				\hline
				12 & 8/20 & 28/58 & 158/282 & 1048/2090 & 0/1638 &  &  & \\
				\hline
				11 & 8/20 & 28/58 & 384/564 & 576/1398  &  &  &  & \\
				\hline
				10 & 8/20 & 34/64 & 244/588 & 0/344  &  &  &  & \\
				\hline
				9 & 8/20 & 48/78 & 112/406 &  &  &  &  & \\
				\hline
				8 & 8/20 & 82/130 & 0/98 &  &  &  &  & \\
				\hline
				7 & 10/22  & 38/98  &  &  &  &  &  & \\
				\hline
				6 & 14/26 & 0/30  &  &  &  &  &  & \\
				\hline
				5 & 8/24  &  &  &  &  &  &  & \\
				\hline
				4 & 0/8  &  &  &  &  &  &  & \\
				\hline
			\end{tabular}
			\caption{Proportions $r_q(n)/s_q(n)$ of right-inextendible binary words of automatic complexity $q$ and length $n$.}
			\label{tab:prop}
		\end{sidewaystable}
		\begin{df}
			We define the set $S_q(n) = \{x \in \{0,1\}^n : A(x) = q\}$ and $s_q(n) = |S_q(n)|$.
		\end{df}

		\begin{df}
			Given an automaton, $G$, whose set of states is $Q$,
			we define a \emph{detour} to be a pair of finite non-trivial sequences of states, $\alpha, \beta \in Q^*$, such that
			$\alpha(0) = \beta(0)$, $\alpha(|\alpha|-1) = \beta(|\beta|-1)$ and $\alpha \neq \beta$.
			We call a detour \emph{minimal} if $\{ \alpha(i) : 0 < i < |\alpha|-1 \} \cap \{ \beta(i) : 0 < i < |\beta|-1 \} = \emptyset$.
		\end{df}

		Consider an automaton with a single cycle (Figure \ref{fig:cycle}).
		Suppose the automaton has $i$ states before the cycle and $\ell$ states after the cycle (which implies that there are $q-(i+\ell)$ states within the cycle).
		We now obtain a formula for the limit of the number of binary words of given complexity $q$.

		\begin{thm}
			\label{thm:total}\label{achilles}
			$s_q$ is eventually constant,
			with limiting value
			\[
				\sum_{\substack{i,\ell\geq 0\\ i+\ell<q}} 2^{(i-1)^+}\cdot [2^{q-(i+\ell)} - Z(q-(i+\ell))]\cdot2^{(\ell-1)^+},
			\]
			where $Z$ was defined in Theorem \ref{thm:upbs}
			and
			\[
				x^+=\max\{x,0\}\textrm{.}
			\]
		\end{thm}
		\begin{proof}
			Consider an arbitrary automaton $G$ with $q$ states.
			There are a finite number of such automata.
			We will prove that unless $G$ has at most one minimal detour,
			there is an $N$ such that, for all $n \geq N$, $G$ cannot accept a unique word of length $n$.

			We begin with the observation that we may assume that $G$ has a unique initial state and a unique accepting state.

			If $G$ has at most one detour, then $G$ has one of the the following forms.

			\begin{center}
				\begin{tabular}{c c c}
					\xy 
						(0,0)*+[o]=<10pt>\hbox{}*\frm{o}="start";
						(50,0)*+[o]=<10pt>\hbox{}*\frm{oo}="end";
						(37.5,0)*+[o]=<5pt>\hbox{}*\frm{*}="p0.a";
						(12.5,0)*+[o]=<5pt>\hbox{}*\frm{*}="p0.b";
						(25.0,12.5)*+{}="p0.c";
						"p0.a";"p0.b"**\crv{"p0.c"} ?>*\dir2{>} ?*_!/-7pt/{_{ k }};
						"start";"p0.b"**\dir{-} ?>*\dir2{>} ?*_!/-7pt/{_{ i }};
						"p0.b";"p0.a"**\dir{-} ?>*\dir2{>} ?*_!/-7pt/{_{ j }};
						"p0.a";"end"**\dir{-} ?>*\dir2{>} ?*_!/-7pt/{_{ \ell }};
					\endxy
					&\hspace{2em}
					&\xy 
						(0,0)*+[o]=<10pt>\hbox{}*\frm{o}="start";
						(50,0)*+[o]=<10pt>\hbox{}*\frm{oo}="end";
						(12.5,0)*+[o]=<5pt>\hbox{}*\frm{*}="p0.a";
						(37.5,0)*+[o]=<5pt>\hbox{}*\frm{*}="p0.b";
						(25.0,12.5)*+{}="p0.c";
						"p0.a";"p0.b"**\crv{"p0.c"} ?>*\dir2{>} ?*_!/7pt/{_{ k }};
						"start";"p0.a"**\dir{-} ?>*\dir2{>} ?*_!/-7pt/{_{ i }};
						"p0.a";"p0.b"**\dir{-} ?>*\dir2{>} ?*_!/-7pt/{_{ j }};
						"p0.b";"end"**\dir{-} ?>*\dir2{>} ?*_!/-7pt/{_{ \ell }};
					\endxy
				\end{tabular}
			\end{center}

			\noindent If $G$ is of the type on the right and $G$ accepts a unique word $\sigma$ of length $n$,
			then any accepting path for $\sigma$ either
			uses the $k$ states that comprise the top path of the detour, or
			uses the $j$ states that comprise the bottom path, but no both.
			Thus, if both $k$ and $j$ are non-zero, there is an automaton with fewer states that accepts only $\sigma$ among all words of length $n$.
			We conclude that in the case of automata with at most one minimal detour, we need only consider ones of the form on the left.

			Now, we consider the possibilities for automata with at least two distinct minimal detours.
			\renewcommand{\theequation}{\arabic{equation}}
			\begin{figure}
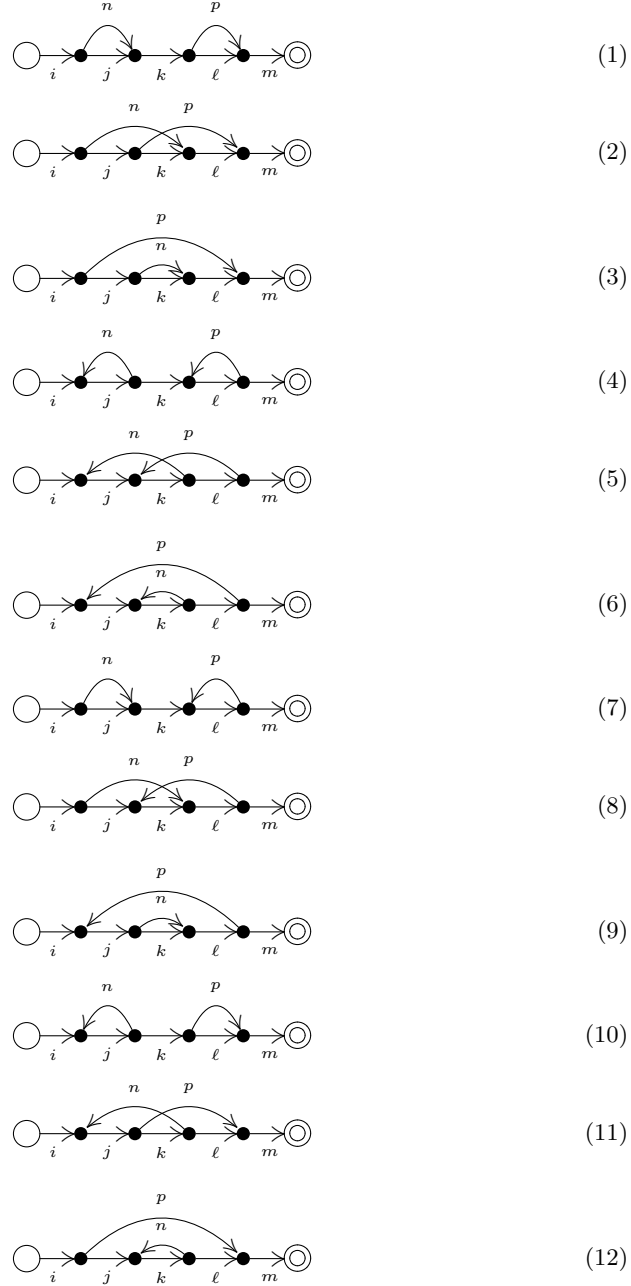

				\centering
				\begin{eqnarray}
					\xy
					(0,0)*+[o]=<10pt>\hbox{}*\frm{o}="start";
					(36,0)*+[o]=<10pt>\hbox{}*\frm{oo}="end";
					(7.2,0)*+[o]=<5pt>\hbox{}*\frm{*}="p0.a";
					(14.4,0)*+[o]=<5pt>\hbox{}*\frm{*}="p0.b";
					(10.8,8)*+{}="p0.c";
					(21.6,0)*+[o]=<5pt>\hbox{}*\frm{*}="p1.a";
					(28.8,0)*+[o]=<5pt>\hbox{}*\frm{*}="p1.b";
					(25.2,8)*+{}="p1.c";
					"p0.a";"p0.b"**\crv{"p0.c"} ?>*\dir2{>} ?*_!/7pt/{_{ n }};
					"p1.a";"p1.b"**\crv{"p1.c"} ?>*\dir2{>} ?*_!/7pt/{_{ p }};
					"start";"p0.a"**\dir{-} ?>*\dir2{>} ?*_!/-7pt/{_{ i }};
					"p0.a";"p0.b"**\dir{-} ?>*\dir2{>} ?*_!/-7pt/{_{ j }};
					"p0.b";"p1.a"**\dir{-} ?>*\dir2{>} ?*_!/-7pt/{_{ k }};
					"p1.a";"p1.b"**\dir{-} ?>*\dir2{>} ?*_!/-7pt/{_{ \ell }};
					"p1.b";"end"**\dir{-} ?>*\dir2{>} ?*_!/-7pt/{_{ m }};
					\endxy
					\label{a}
					\\
					\label{b}
					\xy
					(0,0)*+[o]=<10pt>\hbox{}*\frm{o}="start";
					(36,0)*+[o]=<10pt>\hbox{}*\frm{oo}="end";
					(7.2,0)*+[o]=<5pt>\hbox{}*\frm{*}="p0.a";
					(21.6,0)*+[o]=<5pt>\hbox{}*\frm{*}="p0.b";
					(14.4,7.2)*+{}="p0.c";
					(14.4,0)*+[o]=<5pt>\hbox{}*\frm{*}="p1.a";
					(28.8,0)*+[o]=<5pt>\hbox{}*\frm{*}="p1.b";
					(21.6,7.2)*+{}="p1.c";
					"p0.a";"p0.b"**\crv{"p0.c"} ?>*\dir2{>} ?*_!/7pt/{_{ n }};
					"p1.a";"p1.b"**\crv{"p1.c"} ?>*\dir2{>} ?*_!/7pt/{_{ p }};
					"start";"p0.a"**\dir{-} ?>*\dir2{>} ?*_!/-7pt/{_{ i }};
					"p0.a";"p1.a"**\dir{-} ?>*\dir2{>} ?*_!/-7pt/{_{ j }};
					"p1.a";"p0.b"**\dir{-} ?>*\dir2{>} ?*_!/-7pt/{_{ k }};
					"p0.b";"p1.b"**\dir{-} ?>*\dir2{>} ?*_!/-7pt/{_{ \ell }};
					"p1.b";"end"**\dir{-} ?>*\dir2{>} ?*_!/-7pt/{_{ m }};
					\endxy
					\\
					\label{c}
					\xy
					(0,0)*+[o]=<10pt>\hbox{}*\frm{o}="start";
					(36,0)*+[o]=<10pt>\hbox{}*\frm{oo}="end";
					(14.4,0)*+[o]=<5pt>\hbox{}*\frm{*}="p0.a";
					(21.6,0)*+[o]=<5pt>\hbox{}*\frm{*}="p0.b";
					(18.0,3.6)*+{}="p0.c";
					(7.2,0)*+[o]=<5pt>\hbox{}*\frm{*}="p1.a";
					(28.8,0)*+[o]=<5pt>\hbox{}*\frm{*}="p1.b";
					(18.0,10.8)*+{}="p1.c";
					"p0.a";"p0.b"**\crv{"p0.c"} ?>*\dir2{>} ?*_!/7pt/{_{ n }};
					"p1.a";"p1.b"**\crv{"p1.c"} ?>*\dir2{>} ?*_!/7pt/{_{ p }};
					"start";"p1.a"**\dir{-} ?>*\dir2{>} ?*_!/-7pt/{_{ i }};
					"p1.a";"p0.a"**\dir{-} ?>*\dir2{>} ?*_!/-7pt/{_{ j }};
					"p0.a";"p0.b"**\dir{-} ?>*\dir2{>} ?*_!/-7pt/{_{ k }};
					"p0.b";"p1.b"**\dir{-} ?>*\dir2{>} ?*_!/-7pt/{_{ \ell }};
					"p1.b";"end"**\dir{-} ?>*\dir2{>} ?*_!/-7pt/{_{ m }};
					\endxy
					\\
					\label{d}
					\xy
					(0,0)*+[o]=<10pt>\hbox{}*\frm{o}="start";
					(36,0)*+[o]=<10pt>\hbox{}*\frm{oo}="end";
					(14.4,0)*+[o]=<5pt>\hbox{}*\frm{*}="p0.a";
					(7.2,0)*+[o]=<5pt>\hbox{}*\frm{*}="p0.b";
					(10.8,8)*+{}="p0.c";
					(28.8,0)*+[o]=<5pt>\hbox{}*\frm{*}="p1.a";
					(21.6,0)*+[o]=<5pt>\hbox{}*\frm{*}="p1.b";
					(25.2,8)*+{}="p1.c";
					"p0.a";"p0.b"**\crv{"p0.c"} ?>*\dir2{>} ?*_!/-7pt/{_{ n }};
					"p1.a";"p1.b"**\crv{"p1.c"} ?>*\dir2{>} ?*_!/-7pt/{_{ p }};
					"start";"p0.b"**\dir{-} ?>*\dir2{>} ?*_!/-7pt/{_{ i }};
					"p0.b";"p0.a"**\dir{-} ?>*\dir2{>} ?*_!/-7pt/{_{ j }};
					"p0.a";"p1.b"**\dir{-} ?>*\dir2{>} ?*_!/-7pt/{_{ k }};
					"p1.b";"p1.a"**\dir{-} ?>*\dir2{>} ?*_!/-7pt/{_{ \ell }};
					"p1.a";"end"**\dir{-} ?>*\dir2{>} ?*_!/-7pt/{_{ m }};
					\endxy
					\\
					\label{e}
					\xy
					(0,0)*+[o]=<10pt>\hbox{}*\frm{o}="start";
					(36,0)*+[o]=<10pt>\hbox{}*\frm{oo}="end";
					(21.6,0)*+[o]=<5pt>\hbox{}*\frm{*}="p0.a";
					(7.2,0)*+[o]=<5pt>\hbox{}*\frm{*}="p0.b";
					(14.4,7.2)*+{}="p0.c";
					(28.8,0)*+[o]=<5pt>\hbox{}*\frm{*}="p1.a";
					(14.4,0)*+[o]=<5pt>\hbox{}*\frm{*}="p1.b";
					(21.6,7.2)*+{}="p1.c";
					"p0.a";"p0.b"**\crv{"p0.c"} ?>*\dir2{>} ?*_!/-7pt/{_{ n }};
					"p1.a";"p1.b"**\crv{"p1.c"} ?>*\dir2{>} ?*_!/-7pt/{_{ p }};
					"start";"p0.b"**\dir{-} ?>*\dir2{>} ?*_!/-7pt/{_{ i }};
					"p0.b";"p1.b"**\dir{-} ?>*\dir2{>} ?*_!/-7pt/{_{ j }};
					"p1.b";"p0.a"**\dir{-} ?>*\dir2{>} ?*_!/-7pt/{_{ k }};
					"p0.a";"p1.a"**\dir{-} ?>*\dir2{>} ?*_!/-7pt/{_{ \ell }};
					"p1.a";"end"**\dir{-} ?>*\dir2{>} ?*_!/-7pt/{_{ m }};
					\endxy
					\\
					\label{f}
					\xy
					(0,0)*+[o]=<10pt>\hbox{}*\frm{o}="start";
					(36,0)*+[o]=<10pt>\hbox{}*\frm{oo}="end";
					(21.6,0)*+[o]=<5pt>\hbox{}*\frm{*}="p0.a";
					(14.4,0)*+[o]=<5pt>\hbox{}*\frm{*}="p0.b";
					(18.0,3.6)*+{}="p0.c";
					(28.8,0)*+[o]=<5pt>\hbox{}*\frm{*}="p1.a";
					(7.2,0)*+[o]=<5pt>\hbox{}*\frm{*}="p1.b";
					(18.0,10.8)*+{}="p1.c";
					"p0.a";"p0.b"**\crv{"p0.c"} ?>*\dir2{>} ?*_!/-7pt/{_{ n }};
					"p1.a";"p1.b"**\crv{"p1.c"} ?>*\dir2{>} ?*_!/-7pt/{_{ p }};
					"start";"p1.b"**\dir{-} ?>*\dir2{>} ?*_!/-7pt/{_{ i }};
					"p1.b";"p0.b"**\dir{-} ?>*\dir2{>} ?*_!/-7pt/{_{ j }};
					"p0.b";"p0.a"**\dir{-} ?>*\dir2{>} ?*_!/-7pt/{_{ k }};
					"p0.a";"p1.a"**\dir{-} ?>*\dir2{>} ?*_!/-7pt/{_{ \ell }};
					"p1.a";"end"**\dir{-} ?>*\dir2{>} ?*_!/-7pt/{_{ m }};
					\endxy
					\\
					\label{g}
					\xy
					(0,0)*+[o]=<10pt>\hbox{}*\frm{o}="start";
					(36,0)*+[o]=<10pt>\hbox{}*\frm{oo}="end";
					(7.2,0)*+[o]=<5pt>\hbox{}*\frm{*}="p0.a";
					(14.4,0)*+[o]=<5pt>\hbox{}*\frm{*}="p0.b";
					(10.8,8)*+{}="p0.c";
					(28.8,0)*+[o]=<5pt>\hbox{}*\frm{*}="p1.a";
					(21.6,0)*+[o]=<5pt>\hbox{}*\frm{*}="p1.b";
					(25.2,8)*+{}="p1.c";
					"p0.a";"p0.b"**\crv{"p0.c"} ?>*\dir2{>} ?*_!/7pt/{_{ n }};
					"p1.a";"p1.b"**\crv{"p1.c"} ?>*\dir2{>} ?*_!/-7pt/{_{ p }};
					"start";"p0.a"**\dir{-} ?>*\dir2{>} ?*_!/-7pt/{_{ i }};
					"p0.a";"p0.b"**\dir{-} ?>*\dir2{>} ?*_!/-7pt/{_{ j }};
					"p0.b";"p1.b"**\dir{-} ?>*\dir2{>} ?*_!/-7pt/{_{ k }};
					"p1.b";"p1.a"**\dir{-} ?>*\dir2{>} ?*_!/-7pt/{_{ \ell }};
					"p1.a";"end"**\dir{-} ?>*\dir2{>} ?*_!/-7pt/{_{ m }};
					\endxy
					\\
					\label{h}
					\xy
					(0,0)*+[o]=<10pt>\hbox{}*\frm{o}="start";
					(36,0)*+[o]=<10pt>\hbox{}*\frm{oo}="end";
					(7.2,0)*+[o]=<5pt>\hbox{}*\frm{*}="p0.a";
					(21.6,0)*+[o]=<5pt>\hbox{}*\frm{*}="p0.b";
					(14.4,7.2)*+{}="p0.c";
					(28.8,0)*+[o]=<5pt>\hbox{}*\frm{*}="p1.a";
					(14.4,0)*+[o]=<5pt>\hbox{}*\frm{*}="p1.b";
					(21.6,7.2)*+{}="p1.c";
					"p0.a";"p0.b"**\crv{"p0.c"} ?>*\dir2{>} ?*_!/7pt/{_{ n }};
					"p1.a";"p1.b"**\crv{"p1.c"} ?>*\dir2{>} ?*_!/-7pt/{_{ p }};
					"start";"p0.a"**\dir{-} ?>*\dir2{>} ?*_!/-7pt/{_{ i }};
					"p0.a";"p1.b"**\dir{-} ?>*\dir2{>} ?*_!/-7pt/{_{ j }};
					"p1.b";"p0.b"**\dir{-} ?>*\dir2{>} ?*_!/-7pt/{_{ k }};
					"p0.b";"p1.a"**\dir{-} ?>*\dir2{>} ?*_!/-7pt/{_{ \ell }};
					"p1.a";"end"**\dir{-} ?>*\dir2{>} ?*_!/-7pt/{_{ m }};
					\endxy
					\\
					\label{i}
					\xy
					(0,0)*+[o]=<10pt>\hbox{}*\frm{o}="start";
					(36,0)*+[o]=<10pt>\hbox{}*\frm{oo}="end";
					(14.4,0)*+[o]=<5pt>\hbox{}*\frm{*}="p0.a";
					(21.6,0)*+[o]=<5pt>\hbox{}*\frm{*}="p0.b";
					(18.0,3.6)*+{}="p0.c";
					(28.8,0)*+[o]=<5pt>\hbox{}*\frm{*}="p1.a";
					(7.2,0)*+[o]=<5pt>\hbox{}*\frm{*}="p1.b";
					(18.0,10.8)*+{}="p1.c";
					"p0.a";"p0.b"**\crv{"p0.c"} ?>*\dir2{>} ?*_!/7pt/{_{ n }};
					"p1.a";"p1.b"**\crv{"p1.c"} ?>*\dir2{>} ?*_!/-7pt/{_{ p }};
					"start";"p1.b"**\dir{-} ?>*\dir2{>} ?*_!/-7pt/{_{ i }};
					"p1.b";"p0.a"**\dir{-} ?>*\dir2{>} ?*_!/-7pt/{_{ j }};
					"p0.a";"p0.b"**\dir{-} ?>*\dir2{>} ?*_!/-7pt/{_{ k }};
					"p0.b";"p1.a"**\dir{-} ?>*\dir2{>} ?*_!/-7pt/{_{ \ell }};
					"p1.a";"end"**\dir{-} ?>*\dir2{>} ?*_!/-7pt/{_{ m }};
					\endxy
					\\
					\label{j}
					\xy
					(0,0)*+[o]=<10pt>\hbox{}*\frm{o}="start";
					(36,0)*+[o]=<10pt>\hbox{}*\frm{oo}="end";
					(14.4,0)*+[o]=<5pt>\hbox{}*\frm{*}="p0.a";
					(7.2,0)*+[o]=<5pt>\hbox{}*\frm{*}="p0.b";
					(10.8,8)*+{}="p0.c";
					(21.6,0)*+[o]=<5pt>\hbox{}*\frm{*}="p1.a";
					(28.8,0)*+[o]=<5pt>\hbox{}*\frm{*}="p1.b";
					(25.2,8)*+{}="p1.c";
					"p0.a";"p0.b"**\crv{"p0.c"} ?>*\dir2{>} ?*_!/-7pt/{_{ n }};
					"p1.a";"p1.b"**\crv{"p1.c"} ?>*\dir2{>} ?*_!/7pt/{_{ p }};
					"start";"p0.b"**\dir{-} ?>*\dir2{>} ?*_!/-7pt/{_{ i }};
					"p0.b";"p0.a"**\dir{-} ?>*\dir2{>} ?*_!/-7pt/{_{ j }};
					"p0.a";"p1.a"**\dir{-} ?>*\dir2{>} ?*_!/-7pt/{_{ k }};
					"p1.a";"p1.b"**\dir{-} ?>*\dir2{>} ?*_!/-7pt/{_{ \ell }};
					"p1.b";"end"**\dir{-} ?>*\dir2{>} ?*_!/-7pt/{_{ m }};
					\endxy
					\\
					\label{k}
					\xy
					(0,0)*+[o]=<10pt>\hbox{}*\frm{o}="start";
					(36,0)*+[o]=<10pt>\hbox{}*\frm{oo}="end";
					(21.6,0)*+[o]=<5pt>\hbox{}*\frm{*}="p0.a";
					(7.2,0)*+[o]=<5pt>\hbox{}*\frm{*}="p0.b";
					(14.4,7.2)*+{}="p0.c";
					(14.4,0)*+[o]=<5pt>\hbox{}*\frm{*}="p1.a";
					(28.8,0)*+[o]=<5pt>\hbox{}*\frm{*}="p1.b";
					(21.6,7.2)*+{}="p1.c";
					"p0.a";"p0.b"**\crv{"p0.c"} ?>*\dir2{>} ?*_!/-7pt/{_{ n }};
					"p1.a";"p1.b"**\crv{"p1.c"} ?>*\dir2{>} ?*_!/7pt/{_{ p }};
					"start";"p0.b"**\dir{-} ?>*\dir2{>} ?*_!/-7pt/{_{ i }};
					"p0.b";"p1.a"**\dir{-} ?>*\dir2{>} ?*_!/-7pt/{_{ j }};
					"p1.a";"p0.a"**\dir{-} ?>*\dir2{>} ?*_!/-7pt/{_{ k }};
					"p0.a";"p1.b"**\dir{-} ?>*\dir2{>} ?*_!/-7pt/{_{ \ell }};
					"p1.b";"end"**\dir{-} ?>*\dir2{>} ?*_!/-7pt/{_{ m }};
					\endxy
					\\
					\label{l}
					\xy
					(0,0)*+[o]=<10pt>\hbox{}*\frm{o}="start";
					(36,0)*+[o]=<10pt>\hbox{}*\frm{oo}="end";
					(21.6,0)*+[o]=<5pt>\hbox{}*\frm{*}="p0.a";
					(14.4,0)*+[o]=<5pt>\hbox{}*\frm{*}="p0.b";
					(18.0,3.6)*+{}="p0.c";
					(7.2,0)*+[o]=<5pt>\hbox{}*\frm{*}="p1.a";
					(28.8,0)*+[o]=<5pt>\hbox{}*\frm{*}="p1.b";
					(18.0,10.8)*+{}="p1.c";
					"p0.a";"p0.b"**\crv{"p0.c"} ?>*\dir2{>} ?*_!/-7pt/{_{ n }};
					"p1.a";"p1.b"**\crv{"p1.c"} ?>*\dir2{>} ?*_!/7pt/{_{ p }};
					"start";"p1.a"**\dir{-} ?>*\dir2{>} ?*_!/-7pt/{_{ i }};
					"p1.a";"p0.b"**\dir{-} ?>*\dir2{>} ?*_!/-7pt/{_{ j }};
					"p0.b";"p0.a"**\dir{-} ?>*\dir2{>} ?*_!/-7pt/{_{ k }};
					"p0.a";"p1.b"**\dir{-} ?>*\dir2{>} ?*_!/-7pt/{_{ \ell }};
					"p1.b";"end"**\dir{-} ?>*\dir2{>} ?*_!/-7pt/{_{ m }};
					\endxy
				\end{eqnarray}
				\caption{The possibilities for automata with at least two distinct minimal detours.}\label{capn}
			\end{figure}
			Each of the twelve cases in Figure \ref{capn} falls into one of three cases.

			\begin{enumerate}
				\item On any accepting path, each detour can be used at most once ((\ref{a}), (\ref{b}) and (\ref{c})).
				\item On any accepting path, one of the detours can be used at most once ((\ref{g}), (\ref{h}), (\ref{j}), (\ref{k}) and (\ref{l})).
				\item There are accepting paths that use each of the detours an arbitrary number of times ((\ref{d}), (\ref{e}), (\ref{f}) and (\ref{i})).
			\end{enumerate}

			These further break down as follows:
			\begin{itemize}
				\item (\ref{a}), (\ref{d}), (\ref{g}), (\ref{j}) represent two separated cycles;
				\item (\ref{b}), (\ref{e}), (\ref{h}), (\ref{k}) represent overlapping cycles.
				\item (\ref{c}), (\ref{f}), (\ref{i}), (\ref{l}) represent nested cycles; and
			\end{itemize}

			If $G$ falls into the first case, then
			$\sigma$ is also uniquely accepted among words of length $n$ by an automaton with at most $q$ states and no detours.
			If $G$ falls into the second case, then
			$\sigma$ is uniquely accepted by an automaton with at most $q$ states and at most one detour.
			If $G$ falls into the third category, then
			there are two cycles (although they may have common transitions) which can each be traversed and independent and arbitrary number of times on an accepting path.
			Thus, for large enough $n$, the cycles can be traversed in different orders or different numbers of times and still reach an accepting state,
			thereby violating the requirement that $G$ accept exactly one word of length $n$.

			As an example of the third case, suppose that $G$ is of the type shown in (\ref{i}).
			$G$ has two independent cycles, one of length $p+j+k+\ell$ and the other of length $p+j+n+\ell$.  Let $N = i+a(p+j+k+\ell)+m = i+b(p+j+n+\ell)+m$,
			where $a,b \in \bbn$.
			There are at least two words of length $N$ that $G$ accepts, and
			for any $M \geq N$ such that $G$ accepts a word of length $M$, $G$ must accept at least two words of length $M$.

			In conclusion, we may assume our automata have at most one detour.
			Thus they consist of a chain of states, followed by a single (in general multi-state) cycle, followed by another chain.
			Let $i$ be the number of states before the cycle, $\ell$ the number of states after the cycle, so that $q-(i+\ell)$ if the number of states within the cycle.
			If the bits read within the cycle do not form a necklace, we can reduce the number of states.
			Thus there are $[2^{q-(i+\ell)} - Z(q-(i+\ell))]$ states within the cycle.
			The an upper bound for the total number of binary words with $A_N(x)=q$ is
			\[
				2^{i}\cdot 2^{\ell}\cdot[2^{q-(i+\ell)} - Z(q-(i+\ell))].
			\]
			Let $\xi$ be the bit that advances the automaton from the $i$th state to the $(i+1)$th state (i.e. the transition that takes the automaton into the cycle)
			and $\eta$ be the bit that advances that automaton from the $q-(i+\ell)$th state to the $(i+1)$th state (i.e, the transition that completes the cycle).
			If $\xi=\eta$, then it is possible to create an automaton with fewer states that accepts the same word and no other of length $n$.
			A similar consideration applies upon leaving the cycle.
			Thus, we have
			\[
				2^{(i-1)^+}\cdot [2^{q-(i+\ell)} - Z(q-(i+\ell))]\cdot 2^{(\ell-1)^+}
			\]
			possible words.	

			Finally, to conclude that $s_q(n)$ is eventually constant, note that
			while the single cycle will have to be exited at different points depending on $n$ mod $k$, where $k$ is the length of the main cycle,
			there will always be exactly one value of $n$ mod $k$ and hence exactly one automaton contributed from the cycle and the given ``head'' and ``tail'' words.
			See Figures \ref{fig1}, \ref{fig2}, and \ref{fig3} for illustrations of the cases $q=2, 3, 4$, respectively.
		\end{proof}
		\begin{rem}
			Here is perhaps a simpler view of the classification of detours in Figure \ref{capn}.
			Suppose $A$ is an NFA that uniquely accepts some word.  Now consider
			some shortest directed path $P$ from $q_0$ to the unique final state $q_f$.
			Let us say that an \emph{alternate route} is any
			simple directed path, edge-disjoint from $P$, joining two vertices of $P$.

			Suppose there are two alternate routes, $Q$ and $R$, joining
			$q_i$ and $q_j$, and $q_k$ and $q_l$, respectively. If we do not worry about
			the direction of the paths for the moment, we may assume
			$i \le j$ and $k \le l$.  Then there are three possibilities:
			\begin{enumerate}
				\item $j \le k$:  $Q$ precedes $R$;
				\item $k \le i$ and $j \le l$:  $Q$ encompasses $R$;
				\item $i \le k \le j \le l$:  $Q$ and $R$ overlap.
			\end{enumerate}
			Furthermore, for $Q$ and $R$ one can choose the direction of the edges
			independently.  This gives $3\cdot 4 = 12$ possibilities to consider.
		\end{rem}
		\begin{figure}
			\centering
			\begin{tabular}{c|c|c}
				Count & Regex & Automaton\\
				\hline
				1 & $01^*$ &
					\begin{tikzpicture}[shorten >=1pt,node distance=1.5cm,on grid,auto]
						\node[state,initial] (q_1)   {$q_1$}; 
						\node[state, accepting] (q_2)     [right=of q_1   ] {$q_2$};
						\path[->] 
							(q_1)     edge  node           {$0$}     (q_2)
							(q_2)     edge [loop above] node           {$1$} ();
					\end{tikzpicture}\\
					\hline
					2 & $0^*1$ &
					\begin{tikzpicture}[shorten >=1pt,node distance=1.5cm,on grid,auto]
						\node[state,initial] (q_1)   {$q_1$}; 
						\node[state, accepting] (q_2)     [right=of q_1   ] {$q_2$};
						\path[->] 
							(q_1)     edge  node           {$1$}     (q_2)
							(q_1)     edge [loop above] node           {$0$} ();
					\end{tikzpicture}\\
					\hline
					3 & $(01)^*$ &
					\begin{tabular}{c|c}
						\hline
						$n$ odd & $n$ even \\
					\begin{tikzpicture}[shorten >=1pt,node distance=1.5cm,on grid,auto]
						\node[state,initial] (q_1)   {$q_1$}; 
						\node[state, accepting] (q_2)     [right=of q_1   ] {$q_2$};
						\path[->] 
							(q_1)     edge [bend left]  node           {$0$}     (q_2)
							(q_2)     edge [bend left] node           {$1$} (q_1);
					\end{tikzpicture}
					&
					\begin{tikzpicture}[shorten >=1pt,node distance=1.5cm,on grid,auto]
						\node[state,initial, accepting] (q_1)   {$q_1$}; 
						\node[state] (q_2)     [right=of q_1   ] {$q_2$};
						\path[->] 
							(q_1)     edge [bend left]  node           {$0$}     (q_2)
							(q_2)     edge [bend left] node           {$1$} (q_1);
					\end{tikzpicture}
				\end{tabular}
			\end{tabular}
			\caption{
				The witnessing automata for $\lim_n s_q(n)/2=3$, $q=2$.
				The first two are used at any length $n$, whereas the bottom two are each used only for one value of $n$ mod 2,
				illustrating Theorem \ref{jeffNotesMaybe}.
			}
			\label{fig1}
		\end{figure}
		\begin{figure}
			\centering
			\begin{tabular}{c|c|c}
				Count & Regex & Automata\\
				\hline
				1--5& \makecell[l]{$001^{*}$ (shown)\\ $010^{*}$, $01^{*}0$\\
				$0^{*}10$, $0^{*}11$ }&\begin{tikzpicture}[shorten >=1pt,node distance=1.5cm,on grid,auto]
					\node[state,initial] (q_0)   {$q_0$}; 
					\node[state] (q_1)     [right=of q_0   ] {$q_1$};
					\node[state,accepting] (q_2)     [right=of q_1   ] {$q_2$};
					\path[->] 
						(q_0)     edge node           {$0$}     (q_1)
						(q_1)     edge node           {$0$} (q_2)
						(q_2)     edge [loop above] node           {$1$} ();
				\end{tikzpicture}
				\\
				\hline
				6--8& \makecell[l]{$(001)^*$ (shown)\\
				$(010)^{*}$, $(011)^{*}$} & \begin{tabular}{c|c|c}
				$n\equiv 0\mod 3$ & $n\equiv 1\mod 3$ & $n\equiv 2\mod 3$\\
				\begin{tikzpicture}[shorten >=1pt,node distance=1.5cm,on grid,auto]
					\node[state] (q_1)   {$q_1$}; 
					\node[state] (q_2)     [right=of q_1   ] {$q_2$};
					\node[state, initial below,accepting] (q_0)     [below=of q_1   ] {$q_0$};
					\path[->] 
						(q_0)     edge node           {$0$}     (q_1)
						(q_1)     edge node           {$0$} (q_2)
						(q_2)     edge node           {$1$} (q_0);
				\end{tikzpicture}
				&
				\begin{tikzpicture}[shorten >=1pt,node distance=1.5cm,on grid,auto]
					\node[state,accepting] (q_1)   {$q_1$}; 
					\node[state] (q_2)     [right=of q_1   ] {$q_2$};
					\node[state, initial below] (q_0)     [below=of q_1   ] {$q_0$};
					\path[->] 
						(q_0)     edge node           {$0$}     (q_1)
						(q_1)     edge node           {$0$} (q_2)
						(q_2)     edge node           {$1$} (q_0);
				\end{tikzpicture}
				&
				\begin{tikzpicture}[shorten >=1pt,node distance=1.5cm,on grid,auto]
					\node[state] (q_1)   {$q_1$}; 
					\node[state,accepting] (q_2)     [right=of q_1   ] {$q_2$};
					\node[state, initial below] (q_0)     [below=of q_1   ] {$q_0$};
					\path[->] 
						(q_0)     edge node           {$0$}     (q_1)
						(q_1)     edge node           {$0$} (q_2)
						(q_2)     edge node           {$1$} (q_0);
				\end{tikzpicture}
				\end{tabular}
				\\
				\hline
				9& $0(01)^*$&\begin{tabular}{c|c}
					$n$ odd & $n$ even\\
					\begin{tikzpicture}[shorten >=1pt,node distance=1.5cm,on grid,auto]
						\node[state,accepting] (q_1)   {$q_1$}; 
						\node[state] (q_2)     [right=of q_1   ] {$q_2$};
						\node[state, initial] (q_3)     [below=of q_1   ] {$q_0$};
						\path[->] 
							(q_3)     edge node           {$0$}     (q_1)
							(q_1)     edge [bend left] node           {$0$} (q_2)
							(q_2)     edge [bend left] node           {$1$} (q_1);
					\end{tikzpicture}
					&
					\begin{tikzpicture}[shorten >=1pt,node distance=1.5cm,on grid,auto]
						\node[state] (q_1)   {$q_1$}; 
						\node[state,accepting] (q_2)     [right=of q_1   ] {$q_2$};
						\node[state, initial] (q_3)     [below=of q_1   ] {$q_0$};
						\path[->] 
							(q_3)     edge node           {$0$}     (q_1)
							(q_1)     edge [bend left] node           {$0$} (q_2)
							(q_2)     edge [bend left] node           {$1$} (q_1);
					\end{tikzpicture}
				\end{tabular}\\
				\hline
				10& $(01)^*x$ &
				\begin{tabular}{c|c}
					$n$ odd & $n$ even\\
					\begin{tikzpicture}[shorten >=1pt,node distance=1.5cm,on grid,auto]
						\node[state,initial] (q_1)   {$q_0$}; 
						\node[state] (q_2)     [right=of q_1   ] {$q_1$};
						\node[state, accepting] (q_3)     [below=of q_1   ] {$q_2$};
						\path[->] 
							(q_1)     edge [bend left]  node           {$0$}     (q_2)
							(q_2)     edge [bend left] node           {$1$} (q_1)
							(q_1) edge node {$1$} (q_3);
					\end{tikzpicture}
					&
					\begin{tikzpicture}[shorten >=1pt,node distance=1.5cm,on grid,auto]
						\node[state,initial] (q_1)   {$q_0$}; 
						\node[state] (q_2)     [right=of q_1   ] {$q_1$};
						\node[state, accepting] (q_3)     [below=of q_2   ] {$q_2$};
						\path[->] 
							(q_1)     edge [bend left]  node           {$0$}     (q_2)
							(q_2)     edge [bend left] node           {$1$} (q_1)
							(q_2) edge node {$0$} (q_3);
					\end{tikzpicture}
				\end{tabular}
			\end{tabular}
			\caption{
				Automata and regular expressions witnessing $\lim_n s_q(n)/2=10$ for $q=3$.
				The exponents indicated by $*$ are not necessarily integers (so that for instance $abcd^{1.5}=abcdab$).
				The letter $x$ indicates 0 or 1, chosen so as to break a pattern.
			}
			\label{fig2}
		\end{figure}
		\begin{figure}
			\centering
			\begin{tabular}{c|c|c}
				Count & Regex & Automata\\
				\hline
				1--3 & \makecell[l]{$(001)^{*}x$ (shown),\\ $(010)^{*}x$,\\ $(011)^{*}x$} &
				\begin{tabular}{c|c|c}
					$n\equiv 0\mod 3$ & $n\equiv 1\mod 3$ & $n\equiv 2\mod 3$ \\
					\begin{tikzpicture}[shorten >=1pt,node distance=1.5cm,on grid,auto]
						\node[state,initial above] (q_0)   {$q_0$}; 
						\node[state] (q_2) [below=of q_0] {$q_2$};
						\node[state] (q_1) [right=of q_2] {$q_1$};
						\node[state,accepting] (q_3) [left=of q_2] {$q_3$};
						\path[->] 
							(q_0) edge node {$0$} (q_1)
							(q_1) edge node {$0$} (q_2)
							(q_2) edge node {$1$} (q_0)
							(q_2) edge node {$0$} (q_3);
					\end{tikzpicture}
					&
					\begin{tikzpicture}[shorten >=1pt,node distance=1.5cm,on grid,auto]
						\node[state,initial above] (q_0)   {$q_0$}; 
						\node[state] (q_2) [below=of q_0] {$q_2$};
						\node[state] (q_1) [right=of q_2] {$q_1$};
						\node[state,accepting] (q_3) [right=of q_0] {$q_3$};
						\path[->] 
							(q_0) edge node {$0$} (q_1)
							(q_1) edge node {$0$} (q_2)
							(q_2) edge node {$1$} (q_0)
							(q_0) edge node {$1$} (q_3);
					\end{tikzpicture}
					&
					\begin{tikzpicture}[shorten >=1pt,node distance=1.5cm,on grid,auto]
						\node[state,initial above] (q_0)   {$q_0$}; 
						\node[state] (q_2) [below=of q_0] {$q_2$};
						\node[state] (q_1) [right=of q_2] {$q_1$};
						\node[state,accepting] (q_3) [right=of q_0] {$q_3$};
						\path[->] 
							(q_0) edge node {$0$} (q_1)
							(q_1) edge node {$0$} (q_2)
							(q_2) edge node {$1$} (q_0)
							(q_1) edge node {$1$} (q_3);
					\end{tikzpicture}
				\end{tabular}\\
				\hline
				4--6 & \makecell[l]{$0(001)^{*}$, $0(011)^{*}$,\\ $0(101)^{*}$} & edge followed by cycle of length 3\\
				\hline
				7--8 & \makecell[l]{$(01)^{*}x0$ (shown),\\ $(01)^{*}x1$} &
				\begin{tabular}{c|c}
					$n\equiv 0\mod 2$ & $n\equiv 1\mod 2$\\
					\begin{tikzpicture}[shorten >=1pt,node distance=1.5cm,on grid,auto]
						\node[state,initial above] (q_0)   {$q_0$}; 
						\node[state] (q_1) [right=of q_0] {$q_1$};
						\node[state] (q_2) [below=of q_0] {$q_2$};
						\node[state,accepting] (q_3) [right=of q_2] {$q_3$};
						\path[->] 
							(q_0) edge [bend left] node {$0$} (q_1)
							(q_1) edge [bend left] node {$1$} (q_0)
							(q_0) edge node {$1$} (q_2)
							(q_2) edge node {$0$} (q_3);
					\end{tikzpicture}
					&
					\begin{tikzpicture}[shorten >=1pt,node distance=1.5cm,on grid,auto]
						\node[state,initial above] (q_0)   {$q_0$}; 
						\node[state] (q_1) [right=of q_0] {$q_1$};
						\node[state,accepting] (q_3) [below=of q_0] {$q_3$};
						\node[state] (q_2) [right=of q_3] {$q_2$};
						\path[->] 
							(q_0) edge [bend left] node {$0$} (q_1)
							(q_1) edge [bend left] node {$1$} (q_0)
							(q_1) edge [bend left] node {$0$} (q_2)
							(q_2) edge node {$0$} (q_3);
					\end{tikzpicture}
				\end{tabular}
					\\
				\hline
				9 & $0(01)^{*}x$ & edge followed by cycle of length 2 followed by edge\\
				\hline
				10--15 & \makecell[l]{$(0001)^{*}$ (shown),\\ $(0010)^{*}$,\\ $(0100)^{*}$, $(0011)^{*}$,\\ $(0110)^{*}$, $(0111)^{*}$} & \begin{tabular}{c|c|c|c}
					$n\equiv 0\mod 4$ & $n\equiv 1\mod 4$ & $n\equiv 2\mod 4$ & $n\equiv 3\mod 4$ \\
					\begin{tikzpicture}[shorten >=1pt,node distance=1.5cm,on grid,auto]
						\node[state,initial above,accepting] (q_0)   {$q_0$}; 
						\node[state] (q_1) [right=of q_0] {$q_1$};
						\node[state] (q_2) [below=of q_1] {$q_2$};
						\node[state] (q_3) [left=of q_2] {$q_3$};
						\path[->] 
							(q_0) edge node {$0$} (q_1)
							(q_1) edge node {$0$} (q_2)
							(q_2) edge node {$0$} (q_3)
							(q_3) edge node {$1$} (q_0);
					\end{tikzpicture}
					&
					\begin{tikzpicture}[shorten >=1pt,node distance=1.5cm,on grid,auto]
						\node[state,initial above] (q_0)   {$q_0$}; 
						\node[state,accepting] (q_1) [right=of q_0] {$q_1$};
						\node[state] (q_2) [below=of q_1] {$q_2$};
						\node[state] (q_3) [left=of q_2] {$q_3$};
						\path[->] 
							(q_0) edge node {$0$} (q_1)
							(q_1) edge node {$0$} (q_2)
							(q_2) edge node {$0$} (q_3)
							(q_3) edge node {$1$} (q_0);
					\end{tikzpicture}
					&
					\begin{tikzpicture}[shorten >=1pt,node distance=1.5cm,on grid,auto]
						\node[state,initial above] (q_0)   {$q_0$}; 
						\node[state] (q_1) [right=of q_0] {$q_1$};
						\node[state,accepting] (q_2) [below=of q_1] {$q_2$};
						\node[state] (q_3) [left=of q_2] {$q_3$};
						\path[->] 
							(q_0) edge node {$0$} (q_1)
							(q_1) edge node {$0$} (q_2)
							(q_2) edge node {$0$} (q_3)
							(q_3) edge node {$1$} (q_0);
					\end{tikzpicture}
					&
					\begin{tikzpicture}[shorten >=1pt,node distance=1.5cm,on grid,auto]
						\node[state,initial above] (q_0)   {$q_0$}; 
						\node[state] (q_1) [right=of q_0] {$q_1$};
						\node[state] (q_2) [below=of q_1] {$q_2$};
						\node[state,accepting] (q_3) [left=of q_2] {$q_3$};
						\path[->] 
							(q_0) edge node {$0$} (q_1)
							(q_1) edge node {$0$} (q_2)
							(q_2) edge node {$0$} (q_3)
							(q_3) edge node {$1$} (q_0);
					\end{tikzpicture}
				\end{tabular}\\
				\hline
				16--17 & $00(01)^{*}$, $01(10)^{*}$ & two edges followed by cycle of length 2\\
				\hline
				18--21 & \makecell[l]{$0^{*}101$, $0^{*}110$,\\ $0^{*}111$, $0^{*}100$} & loop followed by chain\\
				\hline
				22--25 & \makecell[l]{$0010^{*}$, $0001^{*}$,\\  $0110^{*}$, $0101^{*}$} & chain followed by loop\\
				\hline
				26--29 & \makecell[l]{$010^{*}1$, $001^{*}0$,\\ $01^{*}00$, $01^{*}01$} & chain of edges with a single loop near middle\\
			\end{tabular}
			\caption{Witnessing automata for $\lim_n s_q(n)/2=29$, $q=4$.
				The exponents indicated by $*$ are not necessarily integers, and the letter $x$ indicates 0 or 1, chosen so as to break a pattern.
			}\label{fig3}
		\end{figure}

		The main proviso to Theorem \ref{thm:total} may be that while \emph{the number of words with given complexity} reaches a limit,
		\emph{the set of witnessing automata} does not quite.
		To wit:
		\begin{thm}\label{jeffNotesMaybe}
			There is a $q$ such that there is no set of automata $M_1,\dots, M_{s}$ such that for all sufficiently large $n$,
			\begin{itemize}
				\item for each $i$ there is some $x$ of length $n$ such that $A_N(x)=q$ and $M_i$ witnesses the inequality $A_N(x)\le q$, and
				\item for all $x$ of length $n$, $A_N(x)=q$ iff the inequality $A_N(x)\le q$ is witnessed by one of the $M_i$.
			\end{itemize}
		\end{thm}
		\begin{proof}
			Let $q=2$. The limiting value of $s_q$ is 6 as witnessed by the patterns: $0^*1$, $01^*$, $(01)^*$.
			However, for $(01)^*$, different states will be the final state depending on the length $n$ mod 2; see Figure \ref{fig1}.
		\end{proof}
		\begin{thm}[Number of right-inextendible words]\label{thm:rtx}
			For $q\ge 1$, define a function $r_q$ by
			\[
				r_q(n) = \#\{x\in \{0,1\}^{n}\mid A_{N}(x)+1=A_N(x0)=A_N(x1)\}.
			\]
			Then $r_q$ is eventually constant, with limiting value
			\[
				\sum_{\substack{i\geq 0,\ell>0\\ i+\ell<q}} 2^{(a-1)^+}\cdot [2^{q-(i+\ell)} - Z(q-(i+\ell))]\cdot2^{\ell-1},
			\]
			where $Z(n)$ refers to the function defined in Theorem 2, and \((x-y)^+:=\max\{(x-y),0\}\).
		\end{thm}
		\begin{proof}
			Let $x$ be a binary word such that its accepting automaton has a single cycle, as in Figure \ref{fig:cycle}.
			As shown in Theorem \ref{thm:total}, we need only consider this particular case.
			Let $\ell$ be the number of states between the cycle and the accepting state of the automaton. 
			\begin{figure}[H]
				\centering
				\includegraphics[scale=0.5]{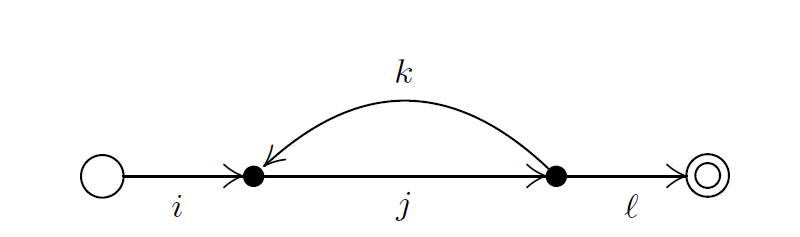}
				\caption{Schematic of an automaton with a single cycle.}
				\label{fig:cycle}
			\end{figure}
			Suppose $\ell=0$.
			Then the accepting state must be one of the states within the cycle.
			Without loss of generality, suppose the path out of the accepting state is triggered by a 0 input.
			Then $x0$ must have the same automatic complexity as $x$, as
			appending 0 to $x$ does not require the addition of any additional states, and $x$ is thus not inextendible.
			Thus, for a word to be inextendible, it is necessary that $\ell>0$.
		\end{proof}
		\begin{thm}\label{372final}
			$s_q(n)$ is eventually bounded by $2^{q-2}\left(\frac{q(q+5)}2 + 1\right)$.
		\end{thm}
		\begin{proof}
			By Theorem \ref{thm:total}, we can upper bound the sum by
			\[
				\sum_{i,\ell\ge 0, i+\ell<q} 2^q = {q+1\choose 2} 2^q.
			\]
			In fact, by considering the four possible truth values for the cases $i=0$, $\ell=0$, we get the upper bound
			\[
				\sum_{i=\ell=0}2^q + \sum_{i\ell=0, i+\ell>0}2^{q-1} + \sum_{i>0,\ell>0}2^{q-2} = 2^q + 2(q-1)2^{q-1} + {q-1\choose 2}2^{q-2}
			\]
			\[
				= 2^{q-2}\left( 4q+{q-1\choose 2}\right) = 2^{q-2}\left(\frac{q(q+5)}2 + 1\right).
			\]
		\end{proof}
		\begin{rem}
			A comparison of $s_q$ with the bound in Theorem \ref{372final} can be done using the computer code in Figure \ref{python}.
			The number in the title of this section was calculated using that Python script and using a table of values of $Z$ from the OEIS database.
			Table \ref{geitost} shows an initial segment of the resulting sequence.
			There we count only words starting with 0, so that the full number would be twice that,
			matching the impression that $\lim_n s_3(n) = 20$ given by Table \ref{tab:prop}.
		\end{rem}
		\begin{table}
			\centering
			\begin{tabular}{r|r||r|r}
				$q$ & $\lim_n s_q(n)/2$ & $q$ & $\lim_n s_q(n)/2$ \\
				\hline
				1 & 1& 21 & 64 594 576\\

				2 & 3&22 & 141 046 655\\
				3 & 10&23 & 306 858 874\\
				4 & 29&24 & 665 342 837\\
				5 & 82&25 & 1 438 134 475\\
				6 & 215&26 & 3 099 548 927\\
				7 & 556&27 & 6 662 442 946\\
				8 & 1 385&28 & 14 285 118 725\\
				9 & 3 391&29 & 30 557 828 119\\
				10 & 8 135&30 & 65 225 030 201\\
				11 & 19 261&31 & 138 937 277 596\\
				12 & 44 963&32 & 295 385 810 819\\
				13 & 103 906&33 & 626 867 939 224\\
				14 & 237 719&34 & 1 328 075 901 017\\
				15 & 539 458&35 & 2 809 126 944 436\\
				16 & 1 214 993&36 & 5 932 793 909 801\\
				17 & 2 718 760&37 & 12 511 847 996 740\\
				18 & 6 047 426&38 & 26 350 575 690 893\\
				19 & 13 380 766&			39 & 55 423 630 773 538\\
				20 & 29 463 632&40 & 116 429 658 505 697\\
			\end{tabular}
			\caption{The number of binary words $0x$ of length $n$ with $A_N(0x)=q$, for sufficiently large $n$.
				The value for $q=7$ is surprisingly small when comparing with Table \ref{tab:prop}.
			}\label{geitost}
		\end{table}

	\bibliographystyle{splncs04}
	\bibliography{ninth-draft}
	\appendix
	\lstset{language=Python}
	\begin{figure}
		\begin{lstlisting}
oeisValues = [
	0, 0, 2, 2, 4, 2, 10, 2, 16, 8, 34, 2, 76, 2, 130, 38, 256, 2,
	568, 2, 1036, 134, 2050, 2, 4336, 32, 8194, 512, 16396, 2, 33814,
	2, 65536, 2054, 131074, 158, 266176, 2, 524290, 8198, 1048816, 2,
	2113462, 2, 4194316, 33272, 8388610, 2, 16842496, 128, 33555424
]# from http://oeis.org/A152061
def Z(n): # number of periodic binary strings of length n
	return oeisValues[n]
def plus(k):
	if k<0:
		return 0
	return k
def limS(q): #limitingNumberOfStringsWithNFAComplexity(q):
	num = 0
	print "."
	for i in range(0, q):
		for l in range(0, q):
			if i+l<q:
				left = 2**(plus(i-1))
				right = 2**(plus(l-1))
				middle = (2**(q-(i+l))-Z(q-(i+l)))
				num += left*middle*right
	return num
def answer(q):
	bound = 2**(q-2)*(1+q*(q+5)/2)
	print "q=" + str(q) + ", " + str(limS(q)),
	print ", bound = " + str(bound) + ", ",
	print str(limS(q)/float(bound))
for q in range(3, len(oeisValues)):
	answer(q)
		\end{lstlisting}
		\caption{Python code which when run hints at the sharpness of Theorem \ref{372final}.}\label{python}
	\end{figure}
\end{document}